\documentclass[11pt]{article}

\usepackage{amsthm}
\usepackage{graphicx} % support the \includegraphics command and options
\usepackage{array} % for better arrays (eg matrices) in maths

\usepackage{amsmath, amssymb, amsfonts, verbatim}
\usepackage{hyphenat,epsfig,subfigure,multirow}

\usepackage[usenames,dvipsnames]{xcolor}
\usepackage[ruled]{algorithm2e}

\usepackage{tcolorbox}
\tcbuselibrary{skins,breakable}
\tcbset{enhanced jigsaw}

\usepackage[compact]{titlesec}

\definecolor{DarkRed}{rgb}{0.5,0.1,0.1}
\definecolor{DarkBlue}{rgb}{0.1,0.1,0.5}

\usepackage{hyperref}
\hypersetup{
colorlinks=true,
pdfnewwindow=true,
citecolor=ForestGreen,
linkcolor=DarkRed,
filecolor=DarkRed,
urlcolor=DarkBlue
}

\usepackage{bm}
\usepackage{url}
\usepackage{xspace}
\usepackage[mathscr]{euscript}

\usepackage{mdframed}

\usepackage[noend]{algpseudocode}
\makeatletter
\def\BState{\State\hskip-\ALG@thistlm}
\makeatother

\usepackage{cite}
\usepackage{enumitem}

\usepackage[margin=1in]{geometry}

\newtheorem{theorem}{Theorem}
\newtheorem{lemma}{Lemma}[section]
\newtheorem{proposition}[lemma]{Proposition}

\newtheorem{claim}[lemma]{Claim}

\newtheorem{definition}{Definition}

\newtheorem{problem}{Problem}
\newtheorem{remark}[lemma]{Remark}

\newtheorem*{claim*}{Claim}
\newtheorem*{proposition*}{Proposition}
\newtheorem*{lemma*}{Lemma}
\newtheorem*{problem*}{Problem}

\newtheorem{mdresult}[theorem]{Theorem}
\newenvironment{Theorem}{\begin{mdframed}[backgroundcolor=lightgray!40,topline=false,rightline=false,leftline=false,bottomline=false,innertopmargin=2pt]\begin{mdresult}}{\end{mdresult}\end{mdframed}}

\newtheorem{mdinvariant}{Invariant}
\newenvironment{invariant}{\begin{mdframed}[hidealllines=false,backgroundcolor=gray!10,innertopmargin=0pt]\begin{mdinvariant}}{\end{mdinvariant}\end{mdframed}}

\newtheorem{question}{Question}

\allowdisplaybreaks

\renewcommand{\qed}{\nobreak \ifvmode \relax \else
      \ifdim\lastskip<1.5em \hskip-\lastskip
      \hskip1.5em plus0em minus0.5em \fi \nobreak
      \vrule height0.75em width0.5em depth0.25em\fi}

\newcommand{\toShrink}{-.20cm}
\newcommand{\toShrinkEnu}{-.2cm}

\newcommand{\Ot}{\ensuremath{\widetilde{O}}}
\newcommand{\eps}{\ensuremath{\varepsilon}}
\newcommand{\Paren}[1]{\Big(#1\Big)}

\newcommand{\bracket}[1]{\left[#1\right]}
\newcommand{\paren}[1]{\ensuremath{\left(#1\right)}\xspace}
\newcommand{\card}[1]{\left\vert{#1}\right\vert}

\newcommand{\set}[1]{\ensuremath{\left\{ #1 \right\}}}

\newcommand{\polylog}{\mbox{\rm  polylog}}

\DeclareMathOperator*{\Exp}{\ensuremath{{\mathbb{E}}}}
\DeclareMathOperator*{\Prob}{\ensuremath{\textnormal{Pr}}}
\renewcommand{\Pr}{\Prob}

\newcommand{\Ex}{\Exp}

\newcommand{\etal}{{\it et al.\,}}

\newcommand{\event}[1]{\ensuremath{{\sf E}_{#1}}}

\newenvironment{tbox}{\begin{tcolorbox}[
		enlarge top by=5pt,
		enlarge bottom by=5pt,
		 breakable,
		 boxsep=0pt,
                  left=4pt,
                  right=4pt,
                  top=10pt,
                  arc=0pt,
                  boxrule=1pt,toprule=1pt,
                  colback=white
                  ]%%
	}
{\end{tcolorbox}}

\renewcommand{\event}{\mathcal{E}}

\renewcommand{\deg}[1]{\ensuremath{\textnormal{\textsf{deg}}_{#1}}}

\renewcommand{\event}[1]{\mathcal{E}(#1)}

\renewcommand{\event}[1]{\ensuremath{\mathcal{E}\paren{#1}}}

\renewcommand{\event}{\ensuremath{\mathcal{E}}\xspace}

\newcommand{\mis}{\ensuremath{\mathcal{M}}}

\title{Fully Dynamic Maximal Independent Set with \\ 
Sublinear in $n$ Update Time}
\author{Sepehr Assadi\footnote{{\small \texttt{sassadi@cis.upenn.edu}.} Supported in part by NSF grant CCF-1617851} \\ University of Pennsylvania
\and Krzysztof Onak\footnote{{\small \texttt{konak@us.ibm.com}.}} \\ IBM Research
\and Baruch Schieber\footnote{{\small \texttt{sbar@us.ibm.com}.}}\\ IBM Research
\and Shay Solomon\footnote{{\small \texttt{solo.shay@gmail.com}.}} \\ IBM Research
}
\date{}

\begin{document}
\maketitle

\thispagestyle{empty}
\begin{abstract}
The first fully dynamic algorithm for maintaining a maximal independent set (MIS) with update time that is \emph{sublinear in the number of edges} was presented recently by the authors of this paper [Assadi \etal STOC'18].
The algorithm is deterministic and its update time is $O(m^{3/4})$, where $m$ is the (dynamically changing) number of edges.
Subsequently, Gupta and Khan and independently Du and Zhang [arXiv, April 2018] presented  deterministic algorithms for dynamic MIS with update times of $O(m^{2/3})$ and
$O(m^{2/3} \sqrt{\log m})$, respectively.
Du and Zhang also gave a randomized algorithm  with update time $\Ot(\sqrt{m})$\footnote{Here and throughout the paper, we use $\Ot$-notation to suppress logarithmic factors, i.e., $\Ot(f) := O(f) \cdot \polylog{(f)}.$}. Moreover, they provided some partial (conditional) hardness results hinting that update time of $m^{1/2-\eps}$, and in particular $n^{1-\eps}$ for $n$-vertex dense graphs, is a natural barrier for this
problem for any constant $\eps >0$, for both deterministic and randomized algorithms that satisfy a certain natural property.

\medskip

In this paper, we break this natural barrier and present the first fully dynamic (randomized) algorithm for maintaining an MIS with update time
that is always \emph{sublinear in the number of vertices}, namely, an $\Ot(\sqrt{n})$ expected amortized update time algorithm. We also show that a simpler variant of our algorithm can already achieve an $\Ot(m^{1/3})$ expected amortized update time, which results in an improved
performance over our $\Ot(\sqrt{n})$ update time algorithm for sufficiently sparse graphs, and breaks the
$m^{1/2}$ barrier of Du and Zhang for all values of $m$.

%%\medskip
%%
%%A key step our algorithms is a novel application of the sample-and-prune technique---used primarily for designing streaming and MapReduce algorithms---to dynamic graph problems, which may be of independent interest.
\end{abstract}

\setcounter{page}{0}
\clearpage
\section{Introduction}

The maximal independent set (MIS) problem is of utmost practical and theoretical importance, primarily since MIS algorithms provide a useful subroutine for locally breaking symmetry between multiple choices.
MIS is often used in the context of graph coloring, as all vertices in an independent set can be assigned the same color.
As another example, Hopcroft and Karp~\cite{HopcroftKarp} gave an algorithm to compute a large bipartite matching (approximating the maximum matching to within a factor arbitrarily close to 1)
by finding maximal independent sets of longer and longer augmenting paths.
In general, the MIS problem has natural connections to various important combinatorial optimization problems; see the celebrated papers of Luby~\cite{Luby86} and Linial~\cite{Linial87} for some of the most basic applications of MIS.
Additional applications of MIS include leader election~\cite{DaumGKN12}, resource allocation~\cite{YuWHL14},  network backbone constructions~\cite{KuhnMW04,JurdzinskiK12}, and sublinear-time approximation algorithms~\cite{NguyenO08}.

The MIS problem has been extensively studied in parallel and distributed settings, following the seminal works of \cite{Luby86,AlonBI86,Linial87}. Surprisingly however, the fundamental problem of \emph{maintaining} an MIS in dynamic graphs received no attention in the literature until the pioneering PODC'16 paper of Censor-Hillel, Haramaty, and Karnin \cite{CHK16}, who developed a \emph{randomized} algorithm for this problem under the oblivious adversarial model\footnote{In the standard \emph{oblivious adversarial model} (cf.\ \cite{CW77b}, \cite{KKM13}), the adversary knows all the edges in the graph and their arrival order, as well as the algorithm to be used,
but is not aware of the random bits used by the algorithm, and so cannot choose updates adaptively in response to the randomly guided choices of the algorithm.}
in \emph{distributed} dynamic networks.
Implementing the distributed algorithm of \cite{CHK16} in the sequential setting requires $\Omega(\Delta)$ update time in \emph{expectation}, where $\Delta$ is a fixed upper bound on the maximum degree in the graph, which may be $\Theta(m)$ in sparse graphs.
Furthermore, it is unclear whether $O(\Delta)$ time is also sufficient for this algorithm, and a naive implementation may incur an update time of $\Theta(m)$, even in expectation, where $m$ is the (dynamically changing) number of edges; see Section~6 of \cite{CHK16} for further details.

We study the MIS problem in (sequential) \emph{dynamic setting}, where the underlying graph evolves over time via edge updates.
A \emph{dynamic graph} is a graph sequence ${\cal G} = (G_0,G_1,\dots,G_M)$ on $n$ fixed vertices,
where the initial graph  is $G_0 = (V,\emptyset)$ and each graph $G_i = (V,E_i)$ is obtained from the previous graph $G_{i-1}$ in the sequence by either adding or deleting a single edge. The work of Censor-Hillel \etal \cite{CHK16} left the following question open:
Can one dynamically maintain an MIS in time significantly lower than it takes to recompute it from scratch following every edge update?

The authors of this paper~\cite{AOSS18} answered this question in the affirmative, presenting the first fully dynamic algorithm for maintaining an MIS with (amortized) update time that is sublinear in the number of edges, namely, $O(\min\{m^{3/4},\Delta\})$. Achieving an update time of $O(\Delta)$ is simple, and the main contribution of~\cite{AOSS18} is in further reducing the update time to $O(m^{3/4})$.
Note that  $O(m^{3/4})$ improves over the simple $O(\Delta) = O(n)$ bound only for sufficiently sparse graphs.

Onak \etal \cite{OSSW18} studied ``uniformly sparse'' graphs, as opposed to the   work by Assadi \etal \cite{AOSS18} that focused on unrestricted sparse graphs.
The ``uniform sparsity'' of the graph is often measured by its \emph{arboricity}~\cite{NashW61,NashW64,Tutte61}:
The arboricity $\alpha$ of a graph $G=(V,E)$ is defined as $\alpha=\max_{U\subset V} \lceil\frac{|E(U)|}{|U|-1}\rceil$, where $E(U)=\left\{(u,v)\in E\mid u,v\in U\right\}$.
A dynamic graph of \emph{arboricity} $\alpha$ is a dynamic graph such that all graphs $G_i$ have arboricity bounded by $\alpha$.
Onak \etal \cite{OSSW18} showed that
for any dynamic $n$-vertex graph of arboricity $\alpha$, an MIS can be maintained with amortized update time $O(\alpha^2\log^2 n)$, which reduces to $O(\log^2 n)$
in bounded arboricity graphs, such as planar graphs and more generally all minor-closed graph classes.
The result of \cite{OSSW18} improves that of \cite{AOSS18}  for all graphs with arboricity bounded by $m^{3/8 - \eps}$, for any constant $\eps > 0$.
Since the arboricity of a general graph cannot exceed $\sqrt{m}$,
this result covers much of the range of possible values for arboricity.
Nonetheless, for general graphs, this update time of $O(\alpha^2\log^2 n)$ is in fact higher than the naive $O(m)$ time needed to compute an MIS from scratch.

Recently, the $O(m^{3/4})$ bound of Assadi \etal \cite{AOSS18} for general graphs was improved to $O(m^{2/3})$ by Gupta and Khan \cite{GuptaK18} and independently to $O(m^{2/3} \sqrt{\log m})$ by Du and Zhang \cite{DZ18}.
All the aforementioned algorithms (besides the distributed algorithm of \cite{CHK16}) are deterministic. Du and Zhang also presented a randomized algorithm under the oblivious adversarial model with an \emph{expected} update time of $O(\sqrt{m} \log^{1.5} m)$;
for dense graphs, this update time reduces to $O(n \log^{1.5} n)$ which is worse than the simple $O(\Delta) = O(n)$ deterministic update time algorithm for this problem.

None of the known algorithms for dynamically maintaining an MIS achieves an update time of $o(n)$ in \emph{dense graphs}.
A recent result of Du and Zhang~\cite{DZ18} partially addresses this lack of progress: they presented an ``imperfect reduction'' from the \emph{Online Boolean Matrix-Vector Multiplication} problem to prove a conditional hardness result for the dynamic MIS problem (see,
e.g.~\cite{HKNS15} for the role of this problem in proving conditional hardness result for dynamic problems).
This result hints that the update time of $m^{1/2-\eps}$ or $n^{1-\eps}$ for any constant $\eps > 0$ maybe of a natural barrier for a large class of deterministic and randomized
algorithms for dynamic MIS that satisfy a certain natural property (see~\cite{DZ18} for exact definition of this property and more details).

This state-of-affairs, namely, the lack of progress on obtaining the update time of $o(n)$ for dynamic MIS in general on one hand, and the partial hardness result hinting that (essentially) $\Omega(n)$ update time might be
a natural barrier for this problem for a large class (but not all) of algorithms on the other hand, raises the following fundamental question:

\begin{question} \label{q1}
%\begin{quote}
Can one maintain a maximal independent set in a dynamically changing graph with update time that is \emph{always} $o(n)$?
%\end{quote}
\end{question}

\subsection{Our contribution}

Our main result is a \emph{positive} resolution of Question~\ref{q1} in a strong sense:

\begin{Theorem}\label{thm1}
Starting from an empty graph on $n$ fixed vertices, an MIS can be maintained over any sequence of edge insertions and deletions in $\Ot(\min\{m^{1/3},\sqrt{n}\})$
amortized update time, where $m$ denotes the dynamic number of edges, and the update time bound holds both in expectation and with high probability%
\footnote{We remark that the high probability
 guarantee holds when the number of updates is sufficiently large; see the formal statement of the results in later sections.}.
\end{Theorem}
The proof of Theorem \ref{thm1} is carried out in three stages. In the first stage we provide a simple randomized algorithm for maintaining an MIS with update time $\Ot(n^{2/3})$;
although we view this as a ``warmup'' result, it already resolves Question \ref{q1}.
In the second stage we generalize this simple algorithm to obtain an update time of $\Ot(m^{1/3})$.
Achieving the $\Ot(\sqrt{n})$ bound is more intricate; we reach this goal by carefully building on the ideas from the $\Ot(n^{2/3})$ and $\Ot(m^{1/3})$-time algorithms.

Finding a maximal independent set is one of the most studied problems in distributed computing. It is thus important to provide an efficient distributed implementation of the proposed sequential dynamic algorithms. While the underlying distributed network is subject to topological updates (particularly edge updates) as in the sequential setting,
the goal in the distributed setting is quite different: Optimizing the (amortized) \emph{round complexity}, \emph{adjustment complexity} and \emph{message complexity} of the distributed algorithm (see, e.g.~\cite{CHK16,AOSS18} for definitions).
Achieving low amortized round and adjustment complexities is typically rather simple, and so the goal is to devise a distributed algorithm whose amortized message complexity matches the update time of the proposed sequential algorithm. This goal was achieved by \cite{AOSS18} and \cite{GuptaK18}. Similarly to \cite{AOSS18, GuptaK18}, our sequential algorithm can also be distributed, achieving an \emph{expected} amortized message complexity of $\Ot(\min\{m^{1/3},\sqrt{n}\})$, in addition to $O(1)$ amortized round and adjustment complexities,  per each update. We omit the details of the distributed implementation of our algorithm as it follows more or less in a straightforward way from our sequential algorithm using the ideas in~\cite{AOSS18}.

\newcommand{\greedy}{\ensuremath{\textnormal{\textsf{Greedy}}}\xspace}

\renewcommand{\deg}[2]{\ensuremath{d_{#1}(#2)}}

\section{Preliminaries}\label{sec:prelim}

\paragraph{Notation.} For a graph $G(V,E)$, $n$ denotes the number of vertices in $V(G) := V$ and $m$ denotes the number of edges in $E(G) := E$. For set $S \subseteq V$, we define $G[S]$ as the \emph{induced} subgraph
of $G$ on vertices in $S$. We further define $N_G(S)$ to be the set of vertices that are neighbor to at least one vertex in $S$ in $G$ (we may drop the subscript $G$ when it is clear from the context). For a vertex $v \in V$,
we define $\deg{G}{v}$ as the degree of $v$ in $G$. Finally, $\Delta(G)$ denotes the maximum degree in $G$.

\paragraph{Greedy MIS.} Maximal independent set problem admits a sequential greedy algorithm.  Let $G(V,E)$ be a graph and $\pi := \pi(V)$ be any ordering of vertices in $G$. $\greedy(G,\pi)$ iterates over vertices in $G$
according to the ordering $\pi$ and adds each vertex to the MIS iff none of its neighbors have already been chosen. It is immediate to verify that this algorithm indeed computes an MIS of $G$ for any ordering $\pi$. Throughout this
paper, we always assume that $\pi$ is the \emph{\underline{lexicographically-first ordering}} of vertices and hence simply write $\greedy(G)$ instead of $\greedy(G,\pi)$.

\subsection{A Deterministic $O(\Delta)$-Update Time Algorithm}\label{sec:det-algorithm}

We use the following simple algorithm for maintaining an MIS deterministically: every vertex maintains a \emph{counter} of number of its neighbors in the MIS, and after any update to the graph,
decides whether it should join or leave the MIS based on this information. Moreover, any vertex that joins or leaves the MIS use $O(\Delta)$ time to update the counter of its neighbors.
While the worst case update time of this algorithm can be quite large for some updates, one can easily prove that \emph{on average}, only $O(\Delta)$ time is needed to process each update,
as was first shown in~\cite{AOSS18} and further strengthened in~\cite{GuptaK18}.
% throughout our proofs: each vertex maintains a counter of its neighbors in the MIS and hence vertices with this counter equal to zero belong to the MIS.
%After each update, we can
\begin{lemma}[\!\!\cite{AOSS18,GuptaK18}]\label{lem:dynamic-delta}
	Starting from any graph $G(V,E)$, a maximal independent set can be maintained \emph{deterministically} over any sequence of $K$ vertex or edge insertions and deletions in $O(m + K \cdot \Delta)$ time
	where $\Delta$ is a fixed bound on the maximum degree in the graph and $m = \card{E}$ is the original number of edges in $G$.
\end{lemma}

\subsection{Sample-and-Prune Technique for Computing an MIS}\label{sec:sample-and-prune}

We also use a simple application of the \emph{sample-and-prune} technique of~\cite{KMVV13} (see also~\cite{LMSV11}) originally introduced in context of streaming and MapReduce algorithms. To our knowledge, the following lemma
was first proved in~\cite{K18} following an approach in~\cite{ACGMW15}. Intuitively speaking, it
asserts that if we sample each vertex of the graph with probability $p$, compute an MIS of the sampled graph, and remove all vertices that are incident to this MIS, the degree of remaining vertices would be $O(p^{-1} \cdot \log{n})$.
For completeness, we present a self-contained proof of this lemma here (we note that our formulation is somewhat different from that of~\cite{K18} and is tailored to our application).

\begin{lemma}[cf.~\cite{K18,ACGMW15}]\label{lem:filtering}
	Fix any $n$-vertex graph $G(V,E)$ and a parameter $p \in [0,1)$. Let $S$ be a collection of vertices chosen by picking each vertex in $V$ independently and with probability $p$.
	Suppose $\mis := \greedy(G[S])$ and $U := V \setminus \paren{\mis \cup N_G(\mis)}$.  Then, with probability $1-1/n^{4}$,
	\[
	\Delta(G[U]) \leq 5p^{-1} \cdot \ln{n}.
	\]
\end{lemma}
\begin{proof}
	Define $\tau := 5p^{-1}\cdot \ln{n}$ and fix any vertex $u$ in the original graph $G$. We prove that with high probability either $u \notin U$ or $\deg{G[U]}{u} \leq \tau$ and then take a union bound on all vertices to conclude the proof.

	We note that the process of computing $\greedy(G[S])$ can be seen as iterating over vertices of $V$ in a lexicographically-first order and skip the vertex if it is incident on $\mis$ (computed so far) and otherwise pick it with probability $p$
	and include it in $\mis$. Let $v_1,\ldots,v_{\deg{G}{u}}$ be the neighbors of $u$ in $G$ ordered accordingly. When processing the vertex $v_i$, if $v_i$ is not already incident on $\mis$ computed so far, the probability
	that we pick $v_i$ to join $\mis$ is exactly $p$. As such, if we encounter at least $\tau$ such vertices in this process, the probability that we do not pick any of them is at most:
	\begin{align*}
		\prod_{i=j}^{\tau} \Pr\paren{\text{$v_{i_j}$ is not chosen $\mid$ $v_{i_j}$ is not incident to the MIS}} = (1-p)^{\tau} \leq \exp\paren{p \cdot 5p^{-1} \cdot \ln{n}} = \frac{1}{n^5}.
	\end{align*}
	As such, we either did not encounter $\tau$ vertices not incident to $\mis$, which implies that $\deg{G[U]}{u} \leq \tau$, or we did, which implies that with probability $1-1/n^5$, $u$ itself is neighbor to some vertex in $\mis$ (as by calculation above, we would pick one of those at least $\tau$ vertices) and hence does not belong to $U$.
	Taking a union bound on all $n$ vertices now finalizes the proof.
\end{proof}

\newcommand{\nei}[1]{\ensuremath{\textnormal{\texttt{neighbors}}[#1]}\xspace}
\newcommand{\mishnei}[1]{\ensuremath{\textnormal{\texttt{MIS-H-neighbors}}[#1]}\xspace}
\newcommand{\lnei}[1]{\ensuremath{\textnormal{\texttt{L-neighbors}}[#1]}\xspace}
\newcommand{\status}[1]{\ensuremath{\textnormal{\texttt{status}}[#1]}\xspace}

\newcommand{\preprocess}{\ensuremath{\textnormal{\textsf{PreProcess}}}\xspace}
\newcommand{\update}{\ensuremath{\textnormal{\textsf{Update}}}\xspace}

\newcommand{\ts}{\ensuremath{t_{start}}}
\renewcommand{\th}{\ensuremath{t_{H}}}
\newcommand{\ti}{\ensuremath{t_{I}}}
\newcommand{\tl}{\ensuremath{t_{L}}}
\newcommand{\td}{\ensuremath{t_{\Delta}}}
\newcommand{\tE}{\ensuremath{t_{edge}}}
\newcommand{\te}{\ensuremath{t_{end}}}

\newcommand{\Deltastar}{\ensuremath{\Delta^{\star}}}

\section{Warmup: An $\Ot(n^{2/3})$-Update Time Algorithm}\label{sec:dynamic-n2/3}

We shall start with a simpler version of our algorithm as a warm-up.

\begin{theorem}\label{thm:dynamic-n2/3}
	Starting from an empty graph on $n$ vertices, a maximal independent set can be maintained via a randomized algorithm over any sequence of $K$ edge insertions and deletions in $O(K \cdot \paren{n\cdot \log{n}}^{2/3})$ time in expectation and
	$O(K \cdot \paren{n\cdot \log{n}}^{2/3} + n^{4/3}\log{n})$ time with high probability\footnote{This in particular implies that when the length of update sequence is $\Omega(n^{2/3})$, the amortized update time is with high probability $O\paren{(n\cdot\log{n})^{2/3})}$.}.
\end{theorem}

The algorithm in Theorem~\ref{thm:dynamic-n2/3} works in \emph{phases}. Each phase starts with a \emph{preprocessing step} in which we initiate the data structure for the algorithm and in particular
compute a partial MIS of the underlying graph with some useful properties (to be specified later).
Next, during each phase, we have the \emph{update step} which processes the updates to the graph until a certain condition (to be defined later) is met, upon which we terminate this phase
and start the next one. We now introduce each step of our algorithm during one phase.

\subsection*{The Preprocessing Step}

The goal in this step is to find a partial MIS of the current graph with the following (informal) properties: $(i)$ it should be ``hard'' for a \emph{non-adaptive oblivious adversary} to ``touch'' vertices of this independent set,
and $(ii)$ maintaining an MIS in the reminder of the graph, i.e., after excluding these vertices and their neighbors from consideration, should be distinctly ``easier''.

In the following, we prove that the sample-and-prune technique introduced in Section~\ref{sec:prelim} can be used to achieve this task (we will pick an exact value for $p$ below later but approximately $p \approx n^{-2/3}$):

\begin{tbox}
	$\preprocess(G,p)$:
	\begin{enumerate}
		\item Let $H$ be a set chosen by picking each vertex in $V(G)$ with probability $p$ independently.
		\item Compute $\mis_H := \greedy(G[H])$.
		\item Return $(H,\mis_H)$.
	\end{enumerate}
\end{tbox}

Throughout this section, we use $\ts$ to denote the time step in which $\preprocess(G,p)$ is computed (hence $G = G_{\ts}$).
We define a partitioning of the vertices of $G_t$ at any time $t \geq \ts$:

\begin{itemize}
	\item $H$: the set of vertices computed by $\preprocess(G_{\ts},p)$ (and \emph{not} $G_t$).
	\item $I_t := N_{G_t}(\mis_H) \setminus H$: the set of vertices incident on $\mis_H$ in the graph $G_t$ that are not in $H$. %We refer to vertices in $I_t$ as \emph{intermediate} vertices.
	\item $L_t := V \setminus (H \cup I_t)$: the set of vertices \emph{not} in $H$ neither incident to $\mis_H$ in the graph $G_t$. %We refer to vertices in $L_t$ as \emph{low priority} vertices.
\end{itemize}
\noindent
It is easy to see that in any time $t \geq \ts$, $(H,I_t,L_t)$ partitions the vertices of the graph.
We emphasize that definition of $H$ is with respect to the time step $\ts$ and graph $G_{\ts}$, while $I_t$ and $L_t$ are defined for the graph $G_t$ for $t \geq \ts$. This means that across time steps $t \geq \ts$,
the set of vertices $H$ is fixed but remaining vertices may move between $I_t$ and $L_t$. We use this partitioning  to define the following key time steps in the execution of the algorithm:
\begin{itemize}
	\item $\th \geq \ts$: the {first} time step $t$ in which $G_t[H] \neq G_{\ts}[H]$ (recall that $H$ and $\mis_H$ were computed with respect to $G_{\ts}$ and not $G_t$).
	\item $\ti \geq \ts$: the first time step $t$ in which the total number of times (since $\ts$) that vertices have moved from $I_s$ to $L_{s+1}$,
	for $s<t$, reaches $2p^{-1}$. 
	\item $\tl \geq \ts$: the {first} time step $t$ in which $\Delta(G_{t}[L_t]) > 5p^{-1} \cdot \ln{n}$.
	\item $\te := \min\set{\th,\ti,\tl,\ts+T}$ where $T:= \frac{1}{6p^2}$: the time step in which we terminate this phase (in other words, if any of the conditions above happen, the phase finishes and the next phase starts).
\end{itemize}
\noindent
By definition above, each phase starts at time step $\ts$ and ends at time step $\te$ and has length at most $T = \frac{1}{6p^2}$. We say that a phase is \emph{successful} iff $\te = \ts + T$.

In the following, we prove that every phase is successful with at least a constant probability (this fact will be used later to argue that the cost of preprocessing steps can be amortized over the large number of
updates between them).

\begin{lemma}\label{lem:is-successful}
	Any given phase is successful, i.e., has $\te = \ts + T$, with probability at least $1/2$.
\end{lemma}
\begin{proof}
The lemma is proved in the following three claims which bound $\th$, $\ti$, and $\tl$, respectively. All claims crucially use the fact the adversary is non-adaptive and oblivious and hence we can fix its updates beforehand.

\begin{claim}\label{clm:t_H}
	$\Pr\paren{\th < \ts + T} \leq \frac{1}{6}$.
\end{claim}
\begin{proof}
	For any $t \geq \ts$, let $e_t := (u_t,v_t)$ denote the
	edge updated by the adversary at time $t$. We consider the randomness in $\preprocess(G_{\ts},p)$. The probability that both $u_t$ and $v_t$ belong to $H$ is exactly $p^2$. For any $t \in [\ts,\ts + T)$,
	define an indicator random variable $X_t$ which is $1$ iff $(u_t,v_t)$ belongs to $G[H]$. Let $X := \sum_{t} X_t$. In order for $G_t[H]$ to no longer be equal to $G_{\ts}[H]$ for some $t \in [\ts,\ts+T)$, at least one of these $T-1$ updates needs to have both
	endpoints in $H$. As such,
	\begin{align*}
		\Pr\paren{\th < \ts + T} \leq \Pr\paren{X \geq 1} \leq \Ex\bracket{X} = (T-1) \cdot {p^2} \leq \frac{1}{6p^2} \cdot p^2 = \frac{1}{6},
	\end{align*}
	where the second inequality is by Markov bound.
\end{proof}

\begin{claim}\label{clm:t_I}
	$\Pr\paren{\ti < \ts + T} \leq \frac{1}{6}$.
\end{claim}
\begin{proof}
	For any $t \geq \ts$, let $e_t := (u_t,v_t)$ denote the edge updated by the adversary at time $t$. By the randomness in $\preprocess(G_{\ts},p)$, the probability that at least one endpoint of $e_t$ belong to $H$ is $2p-p^2 \leq 2p$. For any $t \in [\ts,\ts + T)$, define an indicator random variable $Y_t$ which is $1$ iff at least one of $u_t$ or $v_t$ belong to $H$. Let $Y := \sum_{t} Y_t$.

	The only way a vertex from $I$ moves to $L$ is that an edge incident on this vertex with other endpoint in $\mis_H$ is deleted (and this vertex has no other edge to $\mis_H$ either). For this to happen $2p^{-1}$ times (as in definition of $\ti$),
	we need to have at least $2p^{-1}$ updates in the range $[\ts,\ts+T)$ with at least one endpoint in $H$ (recall that $\mis_H \subseteq H$). As such,
	\begin{align*}
		\Pr\paren{\ti < \ts + T} \leq \Pr\paren{Y \geq 2p^{-1}} \leq \Ex\bracket{Y} \cdot \frac{p}{2} \leq (T-1) \cdot {2p} \cdot \frac{p}{2} \leq \frac{1}{6p^2} \cdot p^2 = \frac{1}{6},
	\end{align*}
	where the second inequality is by Markov bound.
\end{proof}

\begin{claim}\label{clm:t_L}
	$\Pr\paren{\tl < \ts + T \mid \th \geq \ts + T} \leq \frac{1}{n^2}$.
\end{claim}
\begin{proof}
	Fix the graphs $G_t$ for $t \in [\ts,\ts + T)$. Recall that $H$ is a subset of vertices of $G_t$ each chosen independently with probability $p$. Moreover, since $\th \geq \ts + T$ and hence $G_t[H] = G_{\ts}[H]$, we know that
	$\mis_H$ is indeed equal to $\greedy(G_t[H])$ (in addition to $\greedy(G_{\ts}[H])$). As such, by Lemma~\ref{lem:filtering},
	with choice of $S = H$ and $U_t = L_t$, for any graph $G_t$, with probability $1-1/n^4$, we have that $\Delta(G_{t}[L_{t}]) \leq 5p^{-1}\ln{n}$. Taking a union bound on these $\leq n^2$ graphs finalizes the proof.
\end{proof}

By applying union bound to Claims~\ref{clm:t_H},~\ref{clm:t_I}, and~\ref{clm:t_L}, the probability that $\te = \min\set{\th,\ti,\tl} < \ts + T$ is at most $1/6 + 1/6 + 1/n^2 \leq 1/2$, finalizing the proof of Lemma~\ref{lem:is-successful}.
\end{proof}

We conclude this section with the following straightforward lemma.

\begin{lemma}\label{lem:preprocess-runtime}
	$\preprocess(G,p)$ takes $O(m + n)$ time where $m := \card{E(G)}$.
\end{lemma}
\subsection*{The Update Algorithm}

We now describe the update process during each phase. As argued before, each phase continues between time steps $\ts$ and $\te$ where the latter is smaller than or equal to time steps $\th,\ti$ and $\tl$.
As such, by definition of these time steps, we have the following invariant.

\begin{invariant}\label{inv:preprocess}
	At any time step $t \in [\ts,\te)$ inside one phase:
	\begin{enumerate}[label=(\roman*)]
	\item $\mis_H$ is an MIS of the graph $G_t[H]$,
	\item $\Delta(G_t[L_t]) = O(p^{-1} \cdot \log{n})$.
	\end{enumerate}
	Moreover, throughout the phase, at most $O(p^{-1})$ vertices are moved from $I$ to $L$.
\end{invariant}

We note that the first property above is simply because $G_t[H] = G_{\ts}[H]$ as $t < \th$ and hence $\mis_H$ is also an MIS of $G_t[H]$. The second property is by definition of $\tl$ and the last one is by definition of $\ti$.

Our update algorithm simply maintains the graph $G_t[L_t]$ at all time and run the basic deterministic algorithm in Lemma~\ref{lem:dynamic-delta} on $G_t[L_t]$ to maintain an MIS $\mis_{L_t}$ of $G_t[L_t]$. The full MIS maintained
by the dynamic algorithm is then $\mis_H \cup \mis_{L_t}$.

%%We note that to maintain the graph $G_t[L_t]$, we may potentially need to perform vertex insertions and deletions to the underlying graph $G_t[L_t]$ (in addition to standard edge insertions and
%%deletions), and so we crucially use the fact that the algorithm in Lemma~\ref{lem:dynamic-delta} can process even such updates.

We now describe the update algorithm in more details. For any vertex $v \in V$, we maintain whether it currently belongs to $H$, $I_t$, or $L_t$. Additionally, for any vertex in $I_t \cup H$, we maintain a list of its neighbors in $\mis_H$.
Finally, we also maintain the graph $G_t[L_t]$, which involves storing, for each vertex
$v \in L_t$, the set of all of its neighbors in $L_t$. Note that both edges and vertices (as opposed to only edges) may be inserted to or deleted from $G_t[L_t]$ by the algorithm (and as such, we crucially use the fact that the algorithm in
Lemma~\ref{lem:dynamic-delta} can process vertex-updates as well). Fix a time $t \in [\ts,\te]$ and let $e_t = (u_t,v_t)$ be the updated edge. We consider the following cases:

\begin{itemize}[leftmargin=10pt]
	\item \textbf{Case 1.} Updates that cannot impact the partitioning $(H,I_t,L_t)$ of vertices:
	\begin{itemize}
	\item \textbf{Case 1-a.} \emph{Both $u_t$ and $v_t$ belong to $H$}. This update means that $t = \th$ as the graph $G_t[H]$ is updated
	and hence this update concludes this phase (and is processed in the next phase).

	\item \textbf{Case 1-b.} \emph{Both $u_t$ and $v_t$ belong to $I_{t-1}$}. There is nothing to do in this case.

	\item \textbf{Case 1-c.} \emph{Both $u_t$ and $v_t$ belong to $L_{t-1}$}. We need to update the edge $(u_t,v_t)$ in
	the graph $G_t[L_t]$ and pass this edge-update to the algorithm in Lemma~\ref{lem:dynamic-delta} on $G_t[L_t]$.

	\item \textbf{Case 1-d.} \emph{$u_t$ belongs to $I_{t-1}$ and $v_t$ belongs to $L_{t-1}$ (or vice versa)}. There is nothing to do in this case.
	\end{itemize}

	\item \textbf{Case 2.} Updates that can (potentially) change the partitioning $(H,I_t,L_t)$ of vertices:
	\begin{itemize}
	\item \textbf{Case 2-a.} \emph{$u_t$ is in $H$ and $v_t$ is in $I_{t-1}$ (or vice versa)}. If $e_t$ is inserted, the partitioning $(H,I_t,L_t)$ remains the same and there is nothing to do except for updating the list of neighbors of
	$v_t$ in $\mis_H$. However, if $e_t$ is deleted, it might be that $v_t$ needs to be removed from $I_t$ and inserted to $L_t$ instead (if it is no longer incident on $\mis_H$).
	If so, we iterate over all neighbors of $v_t$ and find the ones which are in $L_t$.
	We then insert $v_t$ with all its incident edges to $G_t[L_t]$ and pass this vertex-update to the algorithm in Lemma~\ref{lem:dynamic-delta} on $G_t[L_t]$.

	\item \textbf{Case 2-b.} \emph{$u_t$ is in $H$ and $v_t$ is in $L_{t-1}$ (or vice versa)}. If $e_t$ is deleted, the partitioning $(H,I_t,L_t)$ remains the same and there is nothing to do. However, if $e_t$ is inserted,
	it might be that $v_t$ needs to leave $L_t$ and join $I_t$ (if $u_t$ belongs to $\mis_H$).
	If so, we delete $v_t$ with all its incident edges in $L_t$ from $G_t[L_t]$ and run the algorithm in Lemma~\ref{lem:dynamic-delta} to process this vertex-update in $G_t[L_t]$.

	\end{itemize}
\end{itemize}

\noindent
The cases above cover all possible updates. By the correctness of the deterministic algorithm in Lemma~\ref{lem:dynamic-delta}, $\mis_{L_t}$ is a valid MIS of $G_t[L_t]$. Since all vertices in $I_t$ are incident to some vertex in $\mis_H$, it is
immediate to verify that $\mis_H \cup \mis_{L_t}$ is an MIS of the graph $G$ for any time step $t \in [\ts,\te]$ by Invariant~\ref{inv:preprocess}. It only remains to analyze the running time of the update algorithm.

\begin{lemma}\label{lem:update-time}
	Let $K$ denote the number of updates in a particular phase. The update algorithm maintains an MIS of the graph in this phase in $O(n^2 + p^{-1} \cdot n +  K \cdot p^{-1}\cdot\log{n})$ time.
\end{lemma}
\begin{proof}

	The cost of bookkeeping the data structures in the update algorithm is $O(1)$ per each update. The two main time consuming steps are hence maintaining an MIS in the graph $G_t[L_t]$ and maintaining the graph $G_t[L_t]$ itself.

	The former task, by Lemma~\ref{lem:dynamic-delta}, requires $O(n^2 + K \cdot \Deltastar)$ time in total where $\Deltastar := \max_{t} \Delta(G_t[L_t])$, which by Invariant~\ref{inv:preprocess} is $O(p^{-1}\log{n})$.
	Hence, this part takes $O(n^2 + K \cdot p^{-1}\log{n})$ time in total.

	For the latter task, performing edge updates (in Case 1-c) can be done with $O(1)$ time per each update. Making vertex-deletion updates (in Case 2-b) can also be done in $O(\Deltastar)$ time per update as we only need to iterate over neighbors of the
	updated vertex in $G_t[L_t]$. However, performing the vertex-insertion updates (in Case 2-a) requires iterating over all neighbors of the inserted vertex (in $G_t$ not only $G_t[L_t]$) and hence takes $O(n)$ time.
	Nevertheless, by Invariant~\ref{inv:preprocess}, the total number of such vertex-updates is $O(p^{-1})$ and hence their total running time is $O(p^{-1} \cdot n)$.
\end{proof}

\subsection*{Proof of Theorem~\ref{thm:dynamic-n2/3}}

We are now ready to prove Theorem~\ref{thm:dynamic-n2/3}. The correctness of the algorithm immediately follows from Lemma~\ref{lem:update-time}, hence, it only remains to bound the amortized update time of the algorithm.

Fix a sequence of $K$ updates, and let $P_1,\ldots,P_k$ denote the different phases of the algorithm over this sequence (i.e., each $P_i$ corresponds to the updates
inside one phase). The time spent by the overall algorithm in each phase $i \in [k]$ is $O(n^2)$ in the preprocessing step (by Lemma~\ref{lem:preprocess-runtime}),
and $O(n^2 + p^{-1}\cdot n + \card{P_i} \cdot p^{-1} \cdot \log{n})$ (by Lemma~\ref{lem:update-time}). As such, the total running time is $O(k \cdot (n^2+p^{-1} n) + K \cdot p^{-1} \cdot \log{n})$ (since $\sum_i \card{P_i} = K$). So to finalize the proof, we only need to bound the number
of phases, which is done in the following two lemmas.

\begin{lemma}\label{lem:expected-p}
	$\Ex\bracket{k} = O(K \cdot p^2)$ (the randomness is taken over the coin tosses of the $\preprocess$).
\end{lemma}
\begin{proof}
	Recall that a phase $P_i$ is called successful iff $\card{P_i} = T (= \frac{1}{6p^2})$. The probability that any phase $P_i$ is successful is at least $1/2$ by Lemma~\ref{lem:is-successful}.
	Moreover, since the randomness of $\preprocess$ is independent between any two
	phases, the event that $P_i$ is successful is independent of all previous phases (unless there are no updates left in which case this is going to be the last phase).

	Notice that any successful phase includes $T$ updates and hence we can have at most $K/T$ long phases (even if we assume short phases include no updates). Consider the following randomized process: we have a coin which has at least $1/2$ chance
	of tossing head; how many times in expectation do we need to toss this coin (independently) to see $K/T$ heads? It is immediate to verify that $\Ex\bracket{k}$ is at most this number. It is also standard fact that the expected number of coin tosses in this
	process is $2K/T$. Hence $\Ex\bracket{k} \leq 2K/T = O(K \cdot p^2)$.
\end{proof}

By Lemma~\ref{lem:expected-p}, the expected running time of the algorithm is $O(K \cdot (p^2 \cdot n^2 + p \cdot n) + K \cdot p^{-1}\log{n})$. By picking $p := \frac{(\log{n})^{1/3}}{n^{2/3}}$, we obtain the expected running time of the algorithm is $O(K \cdot \paren{n\cdot \log{n}}^{2/3})$ time, proving
the bound on expected amortized update time in Theorem~\ref{thm:dynamic-n2/3}.

We now prove the high probability bound on the running time.

\begin{lemma}\label{lem:high-probability-p}
	With probability $1-\exp\paren{-K \cdot p^2/10}$, $k  = O(K \cdot p^2)$ (the randomness is taken over the coin tosses of the $\preprocess$). 
\end{lemma}
\begin{proof}
	Recall the coin tossing process described in the proof of Lemma~\ref{lem:expected-p}. Consider the event that among the first $4K/T$ coin tosses, there are at most $K/T$ heads. The probability of this event is at most $\exp\paren{-K/10T}$ by
	a simple application of Chernoff bound.
	On the other hand, the probability of this event is at least equal to the probability that among the first $4K/T$ phases of the algorithm, there are at most $K/T$ long phases. This concludes the proof of first part as we cannot have more than
	$K/T$ long phases among $K$ updates (each long phase ``consumes'' $T$ updates).
\end{proof}

By the choice of $p = \frac{(\log{n})^{1/3}}{n^{2/3}}$, if $K \geq 10 n^{4/3}$, then by Lemma~\ref{lem:high-probability-p}, the running time of the algorithm is $O(K \cdot \paren{n\cdot \log{n}}^{2/3})$, finalizing the proof of this part also.

If however $K < 10n^{4/3}$, we only need \emph{one} successful phase to process all the updates. In this case, since every phase is successful with constant probability, with high probability we only need to consider $O(\log{n})$ phases
before we are done. Moreover, note that when the number of updates is at most $O(n^{4/3})$, the total number of edges in the graph is also $O(n^{4/3})$ only and the preprocessing time takes $O(n^{4/3})$ per each phase as opposed to 
$O(n^2)$. This means that the total running time in this case is at most $O(n^{4/3} \cdot \log{n})$ (for preprocessing) plus $O(K \cdot (n\log{n})^{2/3})$ (time spent inside the phases). This concludes the proof of Theorem~\ref{thm:dynamic-n2/3}.

\newcommand{\Vh}{\ensuremath{V_{\textnormal{\textsf{high}}}}\xspace}
\newcommand{\Vl}{\ensuremath{V_{\textnormal{\textsf{low}}}}\xspace}

\section{An Improved $\Ot(m^{1/3})$-Update Time Algorithm}\label{sec:dynamic-m1/3}

We now show that one can alter the algorithm in Theorem~\ref{thm:dynamic-n2/3} to obtain improved performance for sparser graphs. Formally,

\begin{theorem}\label{thm:dynamic-m1/3}
	Starting from an empty graph, a maximal independent set can be maintained via a randomized algorithm over any sequence of edge insertions and deletions in $O(m^{1/3}\log{m})$ amortized update time both in expectation and
	with high probability, 	where $m$ denotes the dynamic number of edges.
\end{theorem}

The following lemma is a somewhat weaker looking version of Theorem~\ref{thm:dynamic-m1/3}. However, we prove next that this lemma is all we need to prove Theorem~\ref{thm:dynamic-m1/3}.

\begin{lemma}\label{lem:dynamic-m1/3}
	Starting with any arbitrary graph on $m$ edges, a maximal independent set can be maintained via a randomized algorithm over any sequence of $K = \Omega(m)$ edge insertions and deletions in $O(K \cdot m^{1/3}\log{m})$ time in expectation
	and with high probability, as long as the number of edges in the graph remains within a factor $2$ of $m$.
\end{lemma}

We first prove that this lemma implies Theorem~\ref{thm:dynamic-m1/3}. The proof of this part is  standard (see, e.g.~\cite{AOSS18}) and is only provided for completeness.
\begin{proof}[Proof of Theorem~\ref{thm:dynamic-m1/3}]
For simplicity, we define $m = 1$ in case of empty graphs. The idea is to run the algorithm in Lemma~\ref{lem:dynamic-m1/3} until the number of edges deviate from $m$ by a factor more than $2$, upon which, we terminate the algorithm
and restart the process. As the total number of updates is $\Omega(m)$, we can apply Lemma~\ref{lem:dynamic-m1/3} and obtain a bound of $O(m^{1/3}\log{m})$ on the expected amortized update time. Moreover, we can ``charge'' the $O(m)$
time needed to restart the process to the $\Omega(m)$ updates happening in this phase and obtain the final bound.
\end{proof}

The rest of this section is devoted to the proof of Lemma~\ref{lem:dynamic-m1/3}.
The algorithm in Lemma~\ref{lem:dynamic-m1/3} is similar to the one in Theorem~\ref{thm:dynamic-n2/3} and in particular again executes multiple phases each starting by the same preprocessing step (although with change of parameters) followed by
the update algorithm throughout the phase. We now describe the preprocessing step and the update algorithm inside each phase. Recall that throughout this proof, $m$ denotes a $2$-approximation to the number of edges in the graph.

\subsection*{The Preprocessing Step}

Let $\ts$ again denote the first time step in this phase.
The preprocessing step of the new algorithm is exactly as before by running $\preprocess(G_{\ts},p)$ for $p=m^{-1/3}$ (this value of $p$ is different from the one in Section~\ref{sec:dynamic-n2/3} which was $\approx n^{-2/3}$). We define the
partitioning $(H,I_t,L_t)$ of vertices as before. However, we change the stopping criteria of the phase and definition of time steps $\th,\ti,\tl$ as follows :

\begin{itemize}
	\item $\th \geq \ts$: the {first} time step $t$ in which $G_t[H] \neq G_{\ts}[H]$ (recall that $H$ and $\mis_H$ were computed with respect to $G_{\ts}$ and not $G_t$).
	\item $\ti \geq \ts$: the first time step $t$ in which the total number of times (since $\ts$) that vertices have moved from $I_s$ to $L_{s+1}$,
	for $s<t$, reaches $m^{1/3}$. 
	\item $\tl \geq \ts$: the {first} time step $t$ in which $\Delta(G_{t}[L_t]) > 5m^{1/3} \cdot \ln{(m)}$.
	\item $\te := \min\set{\th,\tl,\ts + T}$ where $T := \frac{1}{6} \cdot m^{2/3}$: the time step in which we terminate this phase.
\end{itemize}

We again say that a phase is \emph{successful} if $\te = \ts + T$, i.e., we process $T$ updates in the phase before terminating. Similar to Lemma~\ref{lem:is-successful}, we prove that
each phase is successful with at least a constant probability.

\begin{lemma}\label{lem:is-successful-m}
	Any given phase is successful with probability at least $1/2$.
\end{lemma}
\begin{proof}
The proof is quite similar to Lemma~\ref{lem:is-successful} and is based on the fact that the adversary is non-adaptive and oblivious.

\begin{claim}\label{clm:t_H-m}
	$\Pr\paren{\th < \ts + T} \leq \frac{1}{6}$.
\end{claim}
\begin{proof}
	The proof is identical to Claim~\ref{clm:t_H} by substituting the new values of $p$ and $T$.
\end{proof}

\begin{claim}\label{clm:t_I-m}
	$\Pr\paren{\ti < \ts + T} \leq \frac{1}{6}$.
\end{claim}
\begin{proof}
	Again, the proof is identical to Claim~\ref{clm:t_I} by substituting the new values of $p$ and $T$.
\end{proof}

\begin{claim}\label{clm:t_L-m}
	$\Pr\paren{\tl < \ts + T} \leq \frac{1}{m^2}$.
\end{claim}
\begin{proof}
	Fix the graphs $G_t$ for $t \in [\ts,\ts + T)$ and note that $G_t$ has at most $4m$ vertices with non-zero degree (as number of edges in $G_t$ is at most $2m$)
	and we can ignore vertices with degree zero as they will not affect the following calculation. By Lemma~\ref{lem:filtering},
	with choice of $S = H$ and $U_t = L_t$ for any graph $G_t$ (with at most $4m$ vertices), with probability $1-1/m^4$, $\Delta(G_{t}[L_{t}]) \leq 5p^{-1} \cdot \ln{m} = 5m^{1/3}  \cdot \ln{m}$.
	Taking a union bound on these $\leq m^2$ graphs finalizes the proof.
\end{proof}

By applying union bound to Claims~\ref{clm:t_H-m},~\ref{clm:t_I-m}, and~\ref{clm:t_L-m}, the probability that $\min\set{\th,\tl,\ti} < \ts + T$ is at most $1/6 + 1/6+1/m^2 \leq 1/2$, finalizing the proof of Lemma~\ref{lem:is-successful-m}.
\end{proof}

We conclude this section by noting by Lemma~\ref{lem:preprocess-runtime}, the preprocessing step of this algorithm takes $O(m + n)$ time. However, a simple
trick can reduce the running time to only $O(m)$  as follows.

\begin{lemma}\label{lem:preprocessing-m}
	The preprocessing step of the new algorithm can be implemented in $O(m)$ time.
\end{lemma}
\begin{proof}
	Initially, there are at most $4m$ vertices in the preprocessing step that have non-zero degree. Hence, instead of picking the set $H$ from all of $V$, we only pick it from the vertices with non-zero degree,
	which can be done in $O(m)$ time. Later in the
	algorithm, whenever a new vertex is given an edge in this phase, we toss a coin and decide to add to $H$ with probability $p$ which can be done in $O(1)$ time. We then process this update as before as if this new vertex always belonged to $H$.
	It is immediate to verify that this does not change any part of the algorithm.
\end{proof}

\subsection*{The Update Algorithm}

We now describe the new update algorithm. Firstly, similar to Invariant~\ref{inv:preprocess} in the previous section, here also by definition of each phase, we have that,

\begin{invariant}\label{inv:preprocess-m}
	At any time step $t \in [\ts,\te]$ inside one phase:
	\begin{enumerate}[label=(\roman*)]
	\item $\mis_H$ is an MIS of the graph $G_t[H]$,
	\item $\Delta(G_t[L_t]) = O(m^{1/3}\log{(m)})$.
	\end{enumerate}
	Moreover, throughout the phase, at most $O(m^{1/3})$ vertices are moved from $I$ to $L$.
\end{invariant}

The update algorithm is similar to the one in previous section: we maintain the graph $G_t[L_t]$ and use the algorithm in Lemma~\ref{lem:dynamic-delta} to maintain an MIS $\mis_{L_t}$ in $G_t[L_t]$. The main difference is in how
we maintain the graph $G_t[L_t]$ (the rest is exactly as before). In order to do this, we present a simple data structure.

\paragraph{The Data Structure.} As before, we maintain the list of all neighbors of each vertex, as well as the set $H$, $I_t$, or $L_t$ that it belongs to for each vertex. Clearly, this information can be updated in $O(1)$ time per each update.
In addition to the partition $(H,I_t,L_t)$, we also partition vertices based on their degree in the original graph \emph{at the beginning of the phase}, i.e., in $G_{\ts}$. Specifically, we define $\Vh$ to be the set of vertices with degree at least
$m^{2/3}$ in $G_{\ts}$ and $\Vl := V \setminus \Vh$ to be the remaining vertices. Note that this partitioning is defined with respect to the graph $G_{\ts}$ and does \emph{not} change throughout the phase. We have the following simple claim.

\begin{claim}\label{clm:vh-vl}
	Throughout one phase:
	\begin{enumerate}
	\item $\card{\Vh} = O(m^{1/3})$.
	\item For any vertex $v \in \Vl$ and any graph $G_t$ for $t \geq \ts$, degree of $v$ in $G_t$ is $O(m^{2/3})$.
	\end{enumerate}
\end{claim}
\begin{proof}
	The first is simply because each vertex in $\Vh$ has degree at least $m^{2/3}$ and the total number of edges is at most $2m$. The second part is because the total number of updates inside a phase is at most $\frac{1}{6} \cdot m^{2/3}$ by the
	definition of $\te$ and hence even if they are all incident on a vertex in $\Vl$, the degree of the vertex is at most $\frac{7}{6} \cdot m^{2/3}$, finalizing the proof.
\end{proof}

Finally, for any vertex $v \in \Vh$, we maintain a list of all of its neighbors in $L_t$ as follows: whenever a vertex moves between $I_t$ and $L_t$, it iterates over all vertices in $\Vh$ and inform them of this update. This way, vertices
in $\Vh$ are always aware of their neighborhood in $L_t$. The remaining vertices also have a relatively small degree and hence whenever needed, we could simply iterate over all their neighbors and find the ones in $L_t$.
As a result of this, we have the following invariant.

\begin{invariant}\label{inv:data-structure-m}
	At any time step $t \in [\ts,\te]$ inside one phase after updating $e_t = (u_t,v_t)$:
	\begin{enumerate}[label=(\roman*)]
	\item We can find the list of all neighbors of $u_t$ and $v_t$ that belong to $L_t$ in $O(m^{2/3})$ time.
	\item Updating the data structure after the update takes $O(m^{1/3})$ time.
	\end{enumerate}
\end{invariant}

\begin{proof}
	For vertices in $\Vh$, we have maintained the list of their neighbors explicitly and hence we can directly return this list. For vertices in $\Vl$, we can simply iterate over their $O(m^{2/3})$ neighbors (by Claim~\ref{clm:vh-vl}) and check which one belongs
	to $L_t$ and create the list in $O(m^{2/3})$ time. Finally, the update time is $O(m^{1/3})$ as there are only $O(m^{1/3})$ vertices in $\Vh$ (by Claim~\ref{clm:vh-vl}) and each vertex is only updating these vertices per update.
\end{proof}

\paragraph{Processing Each Update.} We process each update exactly as in the previous section, with the difference that we use Invariant~\ref{inv:data-structure-m}, for maintaining the graph $G_t[L_t]$. To be more specific, in Case 2-a, where a vertex
may be inserted in $G_t[L_t]$, we use the list in Invariant~\ref{inv:data-structure-m}, to find all neighbors of this vertex in $L_t$ and then pass this vertex-update to the algorithm of Lemma~\ref{lem:dynamic-delta} on $G_t[L_t]$. The remaining cases are handled
exactly as before.

The correctness of the algorithm follows as before and we only analyze the running time of the update algorithm.

\begin{lemma}\label{lem:update-time-m}
	Fix any phase and let $K$ denote the number of updates inside this phase. The update algorithm maintains an MIS of the input graph (deterministically) in $O(m+ K \cdot m^{1/3}\cdot\log{m})$ time.
\end{lemma}
\begin{proof}
	By Invariant~\ref{inv:data-structure-m}, updating the data structure takes $O(K \cdot m^{1/3})$ time. Maintaining the MIS in the graph $G_t[L_t]$ also requires $O(m + K \cdot m^{1/3})$
	time by Lemma~\ref{lem:dynamic-delta}. Finally, by Invariant~\ref{inv:data-structure-m}, we can find the neighbors of any updated vertex in $L_t$ in $O(m^{2/3})$ time. Since, the total number of times we need to find
	these neighbors is $O(m^{1/3})$ by Invariant~\ref{inv:preprocess-m} (as we only need this operation when a vertex moves from $I$ to $L$), the total time needed for this part is also $O(m)$, finalizing the proof.
\end{proof}

\subsection*{Proof of Lemma~\ref{lem:dynamic-m1/3}}

The correctness of the algorithm immediately follows from Lemma~\ref{lem:update-time-m}, hence, it only remains to bound the amortized update time of the algorithm.
Fix a sequence of $K$ updates, and let $P_1,\ldots,P_k$ denote the different phases of the algorithm over this sequence (i.e., each $P_i$ corresponds to the updates
inside one phase). The time spent by the overall algorithm in each phase $i \in [k]$ is $O(m)$ in the preprocessing step (by Lemma~\ref{lem:preprocessing-m}),
and $O(m + \card{P_i} \cdot m^{1/3}\log{m})$ (by Lemma~\ref{lem:update-time-m}). As such, the total running time is $O(k \cdot m + K \cdot m^{1/3}\log{m})$ (since $\sum_i \card{P_i} = K$). So to finalize the proof, we only need to bound the number
of phases, which we do in the following lemma.

\begin{lemma}\label{lem:expected-p-m}
	$\Ex\bracket{k} = O(K/m^{2/3})$ (the randomness is taken over the coin tosses of the $\preprocess$).
\end{lemma}
\begin{proof}
	Recall that a phase $P_i$ is called successful iff $\card{P_i} = T (= \frac{1}{6} \cdot m^{2/3})$.
	The probability that any phase $P_i$ is successful is at least $1/2$ by Lemma~\ref{lem:is-successful-m}. Moreover, since the randomness of $\preprocess$ is independent between any two
	phases, the event that $P_i$ is successful is independent of all previous phases (unless there are no updates left in which case this is going to be the last phase).

	Notice that any successful phase includes $T$ updates and hence we can have at most $K/T$ successful phases (even if we assume the other phases include no updates).
	Consider the following randomized process: we have a coin which has at least $1/2$ chance
	of tossing head; how many times in expectation do we need to toss this coin (independently) to see $K/T$ heads? It is immediate to verify that $\Ex\bracket{k}$ is at most this number. It is also standard fact that the expected number of coin tosses in this
	process is $2K/T$. Hence $\Ex\bracket{k} \leq 2K/T = O(K / m^{2/3})$.
\end{proof}

By Lemma~\ref{lem:expected-p-m}, the expected running time of the algorithm is $O(K \cdot m^{1/3} + K \cdot m^{1/3}\ln{m})$, concluding the proof of expectation-bound in Lemma~\ref{lem:dynamic-m1/3}.
The extension to the high probability result now is exactly the same as in Lemma~\ref{lem:high-probability-p}, as $K = \Omega(m) \gg m^{2/3}\log{m}$. This concludes the proof of Lemma~\ref{lem:dynamic-m1/3}.

\newcommand{\mpp}{\ensuremath{\textnormal{\textsf{ModifiedPreProcess}}}\xspace}
\newcommand{\lpp}{\ensuremath{\textnormal{\textsf{LevelPreProcess}}}\xspace}

\newcommand{\bH}{\ensuremath{\overline{H}}}
\newcommand{\tH}{\ensuremath{\widetilde{H}}}

\newcommand{\HH}{\ensuremath{\mathcal{H}}}
\newcommand{\II}{\ensuremath{\mathcal{I}}}
\newcommand{\LL}{\ensuremath{\mathcal{L}}}

\newcommand{\HI}{\ensuremath{H(I)}}
\newcommand{\HL}{\ensuremath{H(L)}}

\section{Main Algorithm: An $\Ot(\sqrt{n})$-Update Time Algorithm}\label{sec:dynamic-n1/2}

We now present our main algorithm for maintaining an MIS in a dynamic graph with $\Ot(\sqrt{n})$ expected amortized update time. 

\begin{theorem}\label{thm:dynamic-n1/2}
	Starting from an empty graph on $n$ vertices, a maximal independent set can be maintained via a randomized algorithm over any sequence of $K$ edge insertions and deletions in 
	$O(K \cdot \sqrt{n} \cdot \log^{2}{n}\cdot\log\log{n})$ time in expectation and in $O(K \cdot \sqrt{n} \cdot \log^{2}{n}\cdot\log\log{n} + n^2\log^{3}{n})$ time with high probability. 
\end{theorem}

The improvement in Theorem~\ref{thm:dynamic-n1/2} over our previous algorithm in Theorem~\ref{thm:dynamic-n2/3} is obtained by using a \emph{nested collection of phases} instead of just one phase.  Let $R := 2\log\log{n}$. 
We maintain $R$ subgraphs of the input graph at any time step of the algorithm, referred to as \emph{level graphs}. For any level $r \in [R]$, we compute and maintain the subgraph at level $r$ in a \emph{level-$r$ phase}. 
A phase as before consists of a preprocessing step, followed by update steps during the phase, and a termination criteria for the phase. Moreover, the phases across different levels are 
nested in a way that a level-1 phase consists of multiple level-2 phases, a level-2 phase contain multiple level-3 phases and so on. We now describe our algorithm in more details starting with the nested family of level graphs.

\subsection*{Level Graphs}

Our approach is based on computing and maintaining a collection of graphs $G^1_t,\ldots,G^R_t$, referred to as level graphs, which are subgraphs of $G_t$ and a collection of independent sets $\mis^1_t,\ldots,\mis^R_t,\mis^{*}_t$.  We maintain the following main invariant in our algorithm (we prove different parts of this invariant in this and the next two sections). 

\begin{invariant}[\textbf{Main Invariant}]\label{inv:main-n1/2}
	At any time step $t$ and for any $r \in [R]$: 
	\begin{enumerate}
	\item\label{inv1} $\mis^1_t \cup \ldots \cup \mis^{R}_t \cup \mis^{*}_t$ is a \emph{maximal independent set} of $G_t$. 
	\item\label{inv2} $\Delta(G^r_t) \leq \Delta_r$ (for parameters $\Delta_r$ to be determined later). 
	\item\label{inv3} $G_r^t$ is maintained explicitly by the algorithm with an adjacency-list access for every vertex.
	\end{enumerate}
\end{invariant}
\noindent
We start by defining the three main collections of vertices of $V(G)$, $\HH_t := \set{H^1_t,\ldots,H^R_t}, \II_t := \set{I^1_t,\ldots,I^R_t}$, and $\LL_t := \set{L^1_t,\ldots,L^R_t}$ used in our algorithm (when clear from the context, or irrelevant, we 
may drop the subscript $t$ from these sets). For simplicity of notation, we also define $H^0_t = I^0_t = \emptyset$ and $L^0_t = V(G)$ for all $t$. We design these sets carefully in the next
section to satisfy the properties below. 

\begin{proposition}\label{prop:first-p}
	At any time step $t$:
\begin{enumerate}
	\item The sets in $\HH_t \cup \II_t$, i.e., $H^1_t,\ldots,H^R_t,I^1_t,\ldots,I^R_t$, are all {pairwise disjoint}. 
	\item The sets in $\LL_t$ are nested, i.e., $L^1_t \supseteq L^2_t \supseteq \ldots \supseteq L^R_t$.
	\item For any fixed $r \in [R]$, $H^r_t,I^r_t,L^r_t \subseteq L^{r-1}_t$ and partition $L^{r-1}_t$.
%	\item For any $r \in [R-1]$, $H^{r+1}_t \subseteq I^r_t \cup L^r_t$. 
\end{enumerate}
\end{proposition}

For any $r \in [R]$, the level-$r$ graph $G^r_t$ is defined as the \emph{induced subgraph} of $G_t$ on $L^r_t$, i.e., $G^r_t := G_t[L^r_t]$. Moreover, $\mis^r_t$ would be chosen carefully 
from the graph $G^{r-1}_t$ such that $\mis^r_t \subseteq H^r_t$. $\mis^*_t$ would also be an MIS of the graph $G^R_t$.  We further have,

\begin{proposition}\label{prop:second-p}
At any time step $t$: 
\begin{enumerate}
	\item For any $r \in [R]$, the independent set $\mis^r_t$ is an MIS of $G_t[H^r_t]$. 
	\item For any $r \in [R]$, $I^r_t$ is incident to some vertex of $\mis^r_t$ and $L^r_t$ has no neighbor in $\mis^r_t$.
%	\item $\mis^*_t$ is a maximal independent set of $G^R_t$. 
\end{enumerate}
\end{proposition}

Before we move on from this section, we show that Proposition~\ref{prop:first-p} and~\ref{prop:second-p} imply the Part-(\ref{inv1}) of Invariant~\ref{inv:main-n1/2}.

\begin{proof}[Proof of Part-(\ref{inv1}) in Invariant~\ref{inv:main-n1/2}]
	Firstly, if follows from Part-(1) and Part-(3) of Proposition~\ref{prop:first-p} that $\HH_t \cup \II_t \cup L^{r}_t$ partitions $V(G) (=L^{0}_t)$. 
	
	By Part-(1) of Proposition~\ref{prop:second-p}, $\mis^1_t$ is an MIS of $G_t[H^1_t]$ and is also incident to all vertices in $I^1_t$. Hence, the only vertices not incident to $\mis^1_t$ are inside $G^1_t$. 
	$\mis^2_t$ is not incident to any vertex of $\mis^1_t$ as $\mis^2_t \subseteq H^2_t \subseteq L^1_t$ and hence by Part-(2) of Proposition~\ref{prop:second-p} are not incident to $\mis^1_t$. We can hence continue
	as before and argue that any vertex not incident to $\mis^2_t$ can only belong to $G^2_t$. Repeating this argument for all $r \in [R]$, we obtain that $\mis^1_t \cup \ldots \cup \mis^R_t$ are all an independent set
	in $G_t$ and moreover, the only vertices not incident to them are in $G^R_t$. Since $\mis^*_t$ is an MIS of $G^R_t$, we obtain that $\mis^1_t \cup \ldots \cup \mis^R_t \cup \mis^*_t$ is an MIS of $G_t$.
\end{proof}

\subsection*{Level-$r$ Phases and Preprocessing Steps}

We now construct the sets $\HH,\II,\LL$, plus the independent sets $\mis^1,\ldots,\mis^R$ from the previous section. These are defined through the notion
of a level-$r$ phases for any $r \in [R]$. Each level-$r$ phase is responsible for maintaining the graph $G^r_t$ and independent set $\mis^r_t$ defined in the previous section.
A level-$r$ phase starts at some time step $\ts^r$ and terminates at some time step $\te^r$ (we define the criteria for terminating a time step later) upon which the next level-$r$ phase starts. 
Both $\ts^r,\te^r \in [\ts^{r-1},\te^{r-1}]$, i.e., any level-$r$ phase happens during a level-$(r-1)$ phase and it is possible that multiple level-$r$ phases start and terminate during a single level-$(r-1)$ phase. 
We now define the process during each phase. 

Pick $R$ probability parameters $p_1,\ldots,p_R \in (0,1)$ such that $p_{r} \geq 2\cdot p_{r-1}$ for all $r > 1$ and $p_{1} \geq \frac{1}{n}$. 
We optimize the values of $p_1,\ldots,p_R$ at the end of the section to obtain the best bound possible from our nested approach. 
Moreover, we define $\Delta_r := \paren{5p_r^{-1} \cdot \ln{n}}$ for all $r \in [R]$ (recall that $\Delta_r$ is used in Part-(\ref{inv2}) of Invariant~\ref{inv:main-n1/2}). 

At the beginning of a level-$r$ phase, we remove all vertices $H^{r},H^{r+1},\ldots,H^{R}$ (similarly for $I$- and $L$-vertices), as well as graphs $G^{r},G^{r+1},\ldots,G^{R}$, and corresponding independent sets $\mis^{r},\ldots,\mis^{R}$,
and $\mis^*$. All these sets and graphs are then redefined through the following \emph{preprocessing step}. 

\begin{tbox}
	$\lpp(r)$ (preprocessing for level-$r$ phases):  
	\begin{enumerate}
		\item Let $t_0 := \ts^r$ denote the current time step. All graphs and sets below are with respect to time $t_0$ and hence we omit this subscript.
	%	\item $G$ is a grph with vertex set $V(G)$ and $V'$ is some superset of $V(G)$, i.e., $V(G) \subseteq V'$.
		\item Let $\tH^r$ be a set chosen by picking each vertex in $V(G)$ independently w.p. $p_r$.
		\item Define $H^r := \tH^r \cap L^{r-1}$ and compute $\mis_{H^r} := \greedy(G[H^r])$.
		\item If $r \leq R$, perform $\lpp(r+1)$.
	\end{enumerate}
\end{tbox}

We note that at first glance it might not be clear that why we pick the set $\tH^r$ from a larger domain in $\lpp$, but then only focus on $H^r$ as the intersection of $\tH^r$ with $L^{r-1}$ (we could have just picked $H^r$ by sampling each vertex in $L^{r-1}$
w.p. $p$). However, we also use the sets $\tH^{r}$ crucially in our algorithm to detect whether the current phase should be terminated or not (for reasons which would become evident shortly) and hence we perform this rather counterintuitive sampling step. 
We now define the sets $\HH_t,\II_t,\LL_t$ plus the independent sets $\mis^1_t,\ldots,\mis^R_t$ for all time steps $t \in [\ts^r,\te^r]$ as follows: 
\begin{itemize}
	\item $\mis^r_t$ is defined to be $\mis_{H^r}$ defined in $\lpp$ throughout the whole phase ($\mis^r_t$ is fixed during a level-$r$ phase).
	\item $H^r_t \in \HH_t$ is defined to be $H^r$ defined in $\lpp$ throughout the whole phase ($H^r_t$ is fixed during a level-$r$ phase).
	\item $I^r_t \in \II_t$ is defined to be any vertex in $L^{r-1}_t$ which is not in $H^r_t$ and is incident to $\mis^r_t$ in the graph $G_t$ ($I^r_t$ can vary during a level-$r$ phase).
	\item $L^r_t \in \LL_t$ is defined to be any vertex in $L^{r-1}_t$ which is not in $H^r_t$ neither in $I^r_t$ in the graph $G_t$ ($L^r_t$ can vary during a level-$r$ phase).
\end{itemize}
\noindent
We now define the termination criterial of a level-$r$ phase using the following time steps.
 
 \begin{itemize}
	\item $\th^r \geq \ts^r$: the {first} time step $t$ where the updated edge $e_t := (u_t,v_t)$ is such that $u_t,v_t \in \tH^1 \cup \ldots \cup \tH^{r}$, 
	and at least one of $u_t$ or $v_t$ belongs to $\tH^{r}$.
	\item $\ti^r \geq \ts^r$: the first time step $t$ in which the total number of times (since $\ts^r$) that vertices in $\tH^1 \cup \ldots \cup \tH^r$ have been incident to an update
	 reaches $p_r^{-1}$. 
	\item $\tl^r \geq \ts^r$: the {first} time step $t$ in which $\Delta(G_{t}[L^r_t]) > \Delta_r$. 
	\item $\te^r := \min\set{\te^{r-1},\th^r,\ti^r,\tl^r,\ts^r+T_r}$ where $T_r:= \frac{1}{24p_r^2}$: the time step in which we terminate this phase (in other words, if any of the conditions above happens, the level-$(r-1)$ that the current level-$r$ phase belongs to 
	terminate, or we simply spend $T_r$ updates in this phase, the phase finishes and the next one starts).  
\end{itemize}

We first prove that by the criteria imposed for terminating each phase, the properties Propositions~\ref{prop:first-p} and~\ref{prop:second-p} are satisfied. 
We start with the simpler proof. 

\begin{proof}[Proof of Proposition~\ref{prop:first-p}]

For simplicity, we drop the subscript $t$ from all sets below. 
\begin{enumerate}
	\item $H^r,I^r,L^r$ are disjoint for each $r \in [R]$ by definition. Moreover, $H^r \cup I^r \subseteq L^{r-1}$, while $H^{r-1},I^{r-1}$ are disjoint from $L^{r-1}$ by definition. This means
	that $H^{r},I^{r}$ are also disjoint from any other set $H^{r'},I^{r'}$ for $r \neq r'$. 
	
	\item Each $L^r$ is defined as a subset of vertices of $L^{r-1}$, hence $L^r \subseteq L^{r-1}$. 
	
	\item The disjointness of $H^r,I^r,L^r$ is by definition. Also, by definition, we have $I^r \cup L^r = L^{r-1} \setminus H^r$, and hence the sets partition $L^{r-1}$ at any time step. \qed
\end{enumerate}

\end{proof}

\begin{proof}[Proof of Proposition~\ref{prop:second-p}]
For simplicity, we drop the subscript $t$ from all sets below. 
	\begin{enumerate}
	\item By definition of $\th^r$, we always terminate a level-$r$ phase and start a new one if the update involved two vertices in $\tH^1 \cup \ldots \cup \tH^r$ with at least one of them in $\tH^r$. 
	As $H^r$ is a subset of $\tH^r$, this means that if an edge with both endpoints in $H^r$ are updated, then we terminate this phase and start a new one. Otherwise, by definition, we have $G_{\ts^r}[H^r] = G_{t}[H^r]$
	for any $t < \th^r$. Since $\mis^r_t = \mis_{H^r}$ was an MIS of $G_{\ts^1}[H^r]$, this means that it is also an MIS of $G_t[H^r]$, proving this part. 
	
	\item This part follows from definition of $\mis^r_t = \mis_{H^r}$ and the sets $I^r$ and $L^r$.	\qed
\end{enumerate}
\end{proof}

We now use these properties to prove Part-(\ref{inv2}) of Invariant~\ref{inv:main-n1/2}.

\begin{proof}[Proof of Part-(\ref{inv2}) of Invariant~\ref{inv:main-n1/2}]
	Recall that $G^r_t := G_t[L^r_t]$. By definition of the time step $\tl$, we always start a new level-$r$ phase whenever $\Delta(G^r_t) > \Delta_r$. As such, throughout the algorithm we always have that $\Delta(G^r_t) \leq \Delta_r$. 
\end{proof}

We also prove the following two auxiliary claims that are used in the rest of the proof. 

\begin{claim}\label{clm:aux-move}
	Let $e_t := (u_t,v_t)$ be an update during a level-$r$ phase after which $v_t$ needs to join or leave $L^r_t$. Then $u_t \in H^{1}_t \cup \ldots \cup H^r_t$. 
\end{claim}
\begin{proof}
	By Proposition~\ref{prop:first-p}, $H^1_t,\ldots,H^r_t,I^1_t,\ldots,I^r_t,L^r_t$ partition $V(G)$. If both $u_t,v_t \in L^r_t$, this update cannot force $v_t$ to leave $L^r_t$. 
	Moreover, if $u_t$ is in $I^1_t \cup \ldots \cup I^r_t$, then deleting or adding this edge does not change the set $L^r_t$ (recall that $L^r_t$ and $I^r_t$ are defined with respect to $H^r_t$ and are both subsets of $L^{r-1}_t$). 
	As such, the only way for $v_t$ to join or leave $L^r_t$ is 
	if $u_t$ belongs to $H^1_t \cup \ldots \cup H^r_t$, finalizing the proof (note that we assumed this update is happening during a phase and hence none of level-$1$ to level-$r$ phases are terminated which naturally change the definition
	of $L^r_t$).
\end{proof}

\begin{claim}\label{clm:aux-n1/2}
	Let $t$ be any time step in $[\ts^r,\th^r)$. Then $\tH^r \cap L^{r-1}_t = H^r$. 
\end{claim}
\begin{proof}
	Recall that $H^r := \tH^r \cap L^{r-1}_{\ts^r}$ and since $L^{r-1}_t$ can vary from $L^{r-1}_{t}$ throughout the phase, a-priori it is not clear that $H^r$ remains the same. However, for $H^r$ to be different from 
	$\tH^r \cap L^{r-1}_t$, a vertex in $\tH^r$, say $v$, should join or leave $L^{r-1}_t$. Consider the first time step $t' \leq t$ such that $v$ did this change and let $(u_{t'},v_{t'})$ be the updated edge at this time step.
	By Claim~\ref{clm:aux-move}, $u_{t'}$ should belong to $H^{1}_{t'} \cup \ldots \cup H^{r-1}_{t'} \subseteq \tH^1_{t'} \cup \ldots \cup \tH^{r-1}_{t'}$. But we also have that $v_{t'} \in \tH^r_{t'}$. This, by definition of $\th^r$
	implies that $t' = \th^r$, contradicting the choice of $t < \th^r$.
\end{proof}

We conclude this part by remarking that definition of the time step $\ti^r$ immediately implies the following invariant. 
\begin{invariant}\label{inv:ti-n1/2}
	The total number of updates during a level-$r$ phase that are incident to some vertex in $\tH^1_t \cup \ldots \cup \tH^r_t$ is $O(p_{r}^{-1})$.
\end{invariant}

\paragraph{Successful Phases.} 
A level-$r$ phase is considered \emph{successul} iff $\te^r = \min\set{\te^{r-1},\ts^r + T_r}$. The following lemma is analogous to Lemma~\ref{lem:is-successful} in Section~\ref{sec:dynamic-n2/3}. 

\begin{lemma}\label{lem:is-successful-n1/2}
	For any $r \in [R]$, any given level-$r$ phase is successful with probability at least $1/2$. 
\end{lemma}
\begin{proof}
	We calculate the probability that $\te^r < \ts^r + T_r$. Recall that the adversary is non-adaptive and oblivious and hence we 
	can fix the updates the adversary. %We start by proving a lower bound on the probability of each of the events in definition of $\te^r$. 
	
\begin{claim}\label{clm:t_H-n1/2}
	$\Pr\paren{\th^r < \ts^r + T_r} \leq \frac{1}{6}$. 
\end{claim}	
\begin{proof}
	For any $t \geq \ts$, let $e_t := (u_t,v_t)$ denote the 
	edge updated by the adversary at time $t$. Define $\event_1(e_t)$ as the event that both $u_t$ and $v_t$ belong to $\tH^1 \cup \ldots \cup \tH^r$ and at least one of them belong to $\tH^r$. Consider 
	the randomness in the choice of $\tH^1 \cup \ldots \cup \tH^r$. We have, 
	\begin{align*}
		\Pr\paren{\event_1(e_t)} \leq 2p_r \cdot (p_1 + \ldots + p_{r}) \leq 2p_r \cdot 2p_r = 4p_r^2,
	\end{align*}
	where we used the fact that $p_r' \geq 2p_{r'-1}$ for all $r' \in [R]$ and hence $p_1 + \ldots + p_r \leq 2p_r$. For any $t \in [\ts,\ts + T)$, define an indicator random variable $X_t$ which is $1$ iff $\event_1(e_t)$ happens. 
	Let $X := \sum_{t} X_t$. As such,  
	\begin{align*}
		\Pr\paren{\th^r < \ts^r + T_r} \leq \Pr\paren{X \geq 1} \leq \Ex\bracket{X} = (T_r-1) \cdot {4p_r^2} \leq \frac{1}{24p_r^2} \cdot 4p_r^2 = \frac{1}{6}, 
	\end{align*}
	where the second inequality is by Markov bound.  
\end{proof}

\begin{claim}\label{clm:t_I-n1/2}
	$\Pr\paren{\ti^r < \ts^r + T_r} \leq \frac{1}{6}$. 
\end{claim}
\begin{proof}
	For any $t \geq \ts$, let $e_t := (u_t,v_t)$ denote the edge updated by the adversary at time $t$. Define $\event_2(e_t)$ as the event that at least one of the endpoints $e_t$ belong to $\tH^1 \cup \ldots \cup \tH^r$. 
	Consider the randomness in the choice of $\tH^1 \cup \ldots \cup \tH^r$. We have, 
	\begin{align*}
		\Pr\paren{\event_2(e_t)} \leq 2(p_1 + \ldots + p_r) \leq 4p_r,
	\end{align*}
	where we used the fact that $p_1 + \ldots + p_r \leq 2p_r$. For any $t \in [\ts,\ts + T)$, define an indicator random variable $Y_t$ which is $1$ iff $\event_2(e_t)$ happens. 
	Let $Y := \sum_{t} Y_t$. We have, 
	\begin{align*}
		\Pr\paren{\ti < \ts + T} \leq \Pr\paren{Y \geq p_r^{-1}} \leq \Ex\bracket{Y} \cdot  p_r \leq (T_r-1) \cdot {4p_r} \cdot p_r \leq \frac{1}{24p_r^2} \cdot 4p_r^2 = \frac{1}{6}, 
	\end{align*}
	where the second inequality is by Markov bound.  
\end{proof}

\begin{claim}\label{clm:t_L-n1/2}
	$\Pr\paren{\tl^r < \ts^r + T_r \mid \th^r \geq \ts^r+T_r} \leq \frac{1}{n^2}$. 
\end{claim}
\begin{proof}
	Let $t_0 = \ts^r$ as in $\lpp$. First consider the graph $G_{t_0}[L^{r-1}_{t_0}]$. The set $H^r$ chosen in $\lpp$ is a set of vertices each chosen with probability $p_r$ from $L^{r-1}_{t_0}$. Hence,
	by Lemma~\ref{lem:filtering}, by choice of $S = H^r$ and $U = L^r_{t_0}$, and since $\mis_{H^r} = \greedy(G[H^r])$, we have that $\Delta(G_{t_0}[L^{r}_{t_0}]) \leq 5p_r^{-1} \ln{n}$. 
	
	Now consider any time step $t > t_0$. Firstly, since $\th^r \geq \ts^r + T_r$, we know that the graph $G_{t_0}[H^r] = G_{t}[H^r]$ and hence $\mis_{H^r}$ is equal to $\greedy(G_{t}[H^r])$ (not only $\greedy(G_{t_0}[H^r])$; this part is 
	identical to the proof of Part-(1) of Proposition~\ref{prop:second-p}).
	The problem with applying the above argument directly for $t$ as well is that the set of vertices in $L^{r-1}_{t}$ may have changed since $L^{r-1}_{t_0}$ and when
	we chose $H^r$. However, consider instead the set $H' := \tH^r \cap L^{r-1}_t$: these are again vertices chosen by picking each vertex of $L^{r-1}_t$ with probability $p_r$ (by definition of $\tH^r$).
	By Claim~\ref{clm:aux-n1/2}, $H' = H^r$ for $t < \th^r$ (which we conditioned on). As such, $H^r = H'$, and we can apply the argument as before and obtain that for any graph $G_t[L^{r-1}_t]$, with probability $1-1/n^4$, $\Delta(G_{t}[L^r_{t}]) \leq 5p_r^{-1}\ln{n}$. Taking a union bound on these $\leq n^2$ graphs finalizes the proof. 
\end{proof}
	By applying union bound to Claims~\ref{clm:t_H-n1/2},~\ref{clm:t_I-n1/2}, and~\ref{clm:t_L-n1/2}, the probability that $\min\set{\th^r,\ti^r,\tl^r} < \ts^r + T_r$ is at most $1/6 + 1/6 + 1/n^2 \leq 1/2$. This concludes
	the proof of Lemma~\ref{lem:is-successful-n1/2}. 
\end{proof}

\subsection*{The Update Algorithm}

We now show how to process the updates during different phases of the algorithm, and prove Part-(\ref{inv3}) of Invariant~\ref{inv:main-n1/2}.  

\paragraph{Processing Updates for a Level-$r$ Phase.} Recall that each level-$r$ phase is mainly responsible for maintaining the graph $G^r_t := G_t[L^r_t]$. We show how to do this in the following. Let $e_t := (u_t,v_t)$ 
be the updated edge. Recall that by Proposition~\ref{prop:first-p}, $H^1_t \cup \ldots \cup H^r_t \cup I^1_t \cup \ldots \cup I^r_t \cup L^r_t$ partitions the set of vertices $V(G)$, and 
$H^{r'}_t \subseteq \tH^{r'}_t$ for all $r' \in [R]$. Finally, we note that we process the updates according to the ordering below and when some updates can be possibly processed according to two or more 
of the cases below, \emph{we always update it according to the first case it appears}. 

\begin{itemize}[leftmargin=10pt]
	\item \textbf{Case 1.} Updates the immediately terminate this phase: 
	\begin{itemize}
		\item \textbf{Case 1-a.} \emph{Both $u_t$ and $v_t$ belong to $\tH^{1}_t \cup \ldots \cup \tH^{r}_t$}. These updates by definition of $\th^r$ either terminate the level-$r$ phase directly, or terminate some level-$r'$ phase for $r' \leq r$, and hence
		indirectly terminate the current level-$r$ phase. These updates are then processed after restarting the level-$r'$ phase (and all phases inside it). 
		
		\item \textbf{Case 1-b.} \emph{Any update that result in time steps $\ti^r,\tl^r,\te^{r-1}$}. These updates are again processed after restarting the current phase and in the next phase. Note that deciding whether an update can result in either of these events can be easily
		detected in $O(1)$ time per each update. 
	\end{itemize}
	\item \textbf{Case 2.} Updates that do not change the set $L^r_t$ (hence do not change vertices of $G^r_t$ but can potentially update its edges):
	\begin{itemize}
	\item \textbf{Case 2-a.} \emph{Both $u_t$ and $v_t$ belong to $I^1_t \cup \ldots \cup I^r_t$}. There is nothing to do in this case. 
	
	\item \textbf{Case 2-b.} \emph{Both $u_t$ and $v_t$ belong to $L^r_t$}. We only need to update the edge in $(u_t,v_t)$ in the graph $G^r_t := G_t[L^r_t]$ which can be done in $O(1)$ time trivially. 
		
	\item \textbf{Case 2-c.} \emph{$u_t$ belongs to $I^1_t \cup \ldots \cup I^r_t$ and $v_t$ belongs to $L^r_t$ (or vice versa)}. There is nothing to do in this case either. 
	\end{itemize}
	
	\item \textbf{Case 3.} Updates that can (potentially) change the set $L^r_t$ (and hence the vertices of $G^r_t$); recall that by Claim~\ref{clm:aux-move}, one endpoint of any such update needs to be in $H^1_t \cup \ldots \cup H^r_t$:
	\begin{itemize}
	\item \textbf{Case 3-a.} \emph{$u_t$ is in $H^1_t \cup \ldots \cup H^r_t$ and $v_t$ is in $I^1_t \cup \ldots \cup I^r_{t}$ (or vice versa)}. If $e_t$ is inserted, no set needs to be changed. 
	However, if $e_t$ is deleted, it might be that $v_t$ needs to be removed from $I^{r'}_t$ and be inserted in either $I^{r''}_t$ for some $r' \leq r'' \leq r$ or to $L^r_t$. 
	If it is to be inserted in $I^{r''}_t$, it necessarily means that it also needs to be inserted to the set $L^{r''-1}_t$ and hence we process this update at a level-$r''$ phase as well and that phase then informs the next level phase in case
	it needs to also add $v_t$ to its corresponding $L$-set and so on. As such, without loss of generality, in the current phase, we can focus on the case when $v_t$ needs to be inserted to $L^r_t$. 
	
	To do this, we iterate over all neighbors of $v_t$ in the graph $G^{r-1}_t$ and find all the ones that also belong to the set $L^{r}_t$ (recall that $L^{r}_t \subseteq L^{r-1}_t$ and $L^{r-1}_t$ is the vertex-set of $G^{r-1}_t$). 
	This takes $O(\Delta_{r-1})$ time as maximum degree of $G^{r-1}_t$ is at most $\Delta_{r-1}$ by Part-(\ref{inv2}) of Invariant~\ref{inv:main-n1/2}. Here, we also assumed inductively that Part-(\ref{inv3}) of Invariant~\ref{inv:main-n1/2}, 
	holds for graphs $G^{1}_t,\ldots,G^{r-1}_t$. We then insert this vertex to $G^{r}_t$ and update the adjacency-list of all its neighbors
	in the graph $G^{r}_t$ in $O(\Delta_r)$ time. Finally, we pass this update to the next level phase to process (updates of this form are passed from higher level phases to lower level phases). 
	
	\item \textbf{Case 3-b.} \emph{$u_t$ is in $H^1_t \cup \ldots \cup H^r_t$ and $v_t$ is in $L^{r}_{t}$ (or vice versa)}. If $e_t$ is deleted, no sets need to be changed. However, if $e_t$ is inserted, 
	it might be that $v_t$ needs to leave $L^r_t$ and join $I^{r'}_t$ for some $r' \leq r$. We first delete $v_t$ with all its incident edges from $G^r_t$ using the adjacency-list representation we maintained for this graph. This takes
	$O(\Delta_r)$ time as the maximum degree of $G^r_t$ is at most $\Delta_r$ by Part-(\ref{inv2}) of Invariant~\ref{inv:main-n1/2}. If this vertex needs to be inserted to $I^r_t$ we do so, otherwise there is nothing to do (note that,
	this update is being processed by \emph{all} level-$r'$ phases for $r' \leq r$ and the corresponding level that needs to insert $v_t$ to $I^{r'}_t$ would do so). 
	\end{itemize}
\end{itemize}
\noindent
One can verify that cases above contain all possible updates. This immediately proves Part-(\ref{inv3}) of Invariant~\ref{inv:main-n1/2}. 

\textbf{Processing Updates to Maintain $\mis^*_t$.} Recall that we also need to maintain $\mis^*_t$ which is an MIS of the graph $G^R_t$. To do this, we simply run the deterministic algorithm of Lemma~\ref{lem:dynamic-delta} on the graph
$G^{R}_t$ which we are explicitly maintaining by  Part-(\ref{inv3}) of Invariant~\ref{inv:main-n1/2}. As this deterministic algorithm can handle vertex-insertions and deletions as well as edge insertions and deletions, $\mis^*_t$ would 
indeed be an MIS of $G^R_t$ and this requires $O(\Delta_R)$ amortized update time as maximum degree of $G^R_t$ is $\Delta_R$ by Part-(\ref{inv2}) of Invariant~\ref{inv:main-n1/2}. 

We now bound the total running time of the algorithm responsible for each phase, as well as the one needed for maintaining $\mis^*_t$. 

\begin{lemma}\label{lem:update-time-phase}
	Let $K$ denote the number of updates in a particular level-$r$ phase. The update algorithm for the level-$r$ phase maintains the independent set $\mis^r_t$ and graph $G^r_t$ (deterministically) in $O(n \cdot \Delta_{r-1} + K)$ time.  
\end{lemma}
\begin{proof}

	The cost of bookkeeping the data structures in the update algorithm is $O(1)$ per each update. The two main time consuming steps are hence the preprocessing done at the beginning of the level-$r$ phase and 
	the cost of maintaining the graph $G^r_t$ throughout the phase. 
	
	The preprocessing algorithm takes linear time in the graph it processes. As it is performed over $G^{r-1}_{\ts^r}$ and maximum degree of $G^{r-1}_t \leq \Delta_{r-1}$ throughout (by Invariant~\ref{inv:main-n1/2}), the preprocessing
	step of a level-$r$ phase takes $O(n \cdot \Delta_{r-1})$ time. 
	
	For the latter task, performing all updates except for Case $3$ can be done in $O(1)$ time per each update, while Case $3$ updates require $O(\Delta_{r-1})$ time per update as argued above. However note that by Claim~\ref{clm:aux-move},
	any Case $3$ update necessarily contains a vertex in $H^1_t \cup \ldots \cup H^r_t \subseteq \tH^1_t \cup \ldots \cup \tH^r_t$. By Invariant~\ref{inv:ti-n1/2}, the total number of such updates during a level-$r$ phase
	is at most $p_r^{-1}$. As such, the total time needed to process Case $3$ phases is $O(p_r^{-1} \cdot \Delta_{r-1})$ which is at most $O(n \cdot \Delta_{r-1})$ as $p_r^{-1} \leq p_1^{-1} \leq n$. 
\end{proof}

\begin{lemma}\label{lem:update-time-mis}
	Let $K$ denote the number of updates in a particular level-$R$ phase. The update algorithm maintains an MIS $\mis^*_t$ in $G^R_t$ (deterministically) in $O(n \cdot \Delta_R + K \cdot \Delta_{R})$ time.
\end{lemma}
\begin{proof}
	Follows from Lemma~\ref{lem:dynamic-delta} as by Invariant~\ref{inv:main-n1/2}, maximum degree of $G^R_t$ is at most $\Delta_R$ and we only ``start'' the deterministic algorithm in Lemma~\ref{lem:dynamic-delta}, once 
	per each level-$R$ phase.  
\end{proof}

\subsection*{Proof of Theorem~\ref{thm:dynamic-n1/2}}

We are now ready to prove Theorem~\ref{thm:dynamic-n1/2}. The correctness of the algorithm immediately follows from Lemmas~\ref{lem:update-time-phase} and~\ref{lem:update-time-mis} and Part-(\ref{inv1}) of Invariant~\ref{inv:main-n1/2}, 
hence, it only remains to bound the amortized 
update time of the algorithm. 

Fix a sequence of $K$ updates and for any $r \in [R]$, let $P^r_1,\ldots,P^r_{k_r}$ denote the different phases of the algorithm over this sequence (i.e., each $P^r_{i}$ corresponds to the updates
inside one level-$r$ phase). We compute the time spent by the overall algorithm in level-$r$ phase, as well as the algorithm for maintaining $\mis^*_t$ separately. 

\paragraph{Total Time Spent Across All Level-$r$ Phases.} By Lemma~\ref{lem:update-time-phase}, the total time spent across all level-$r$ phases is $O(k_r \cdot n \cdot \Delta_{r-1} + K)$ as $K = \sum_{i} \card{P^r_i}$. 
Hence, we only need to upper bound $k_r$. 

\begin{lemma}\label{lem:expected-k_r}
	For any $r \in [R]$, $\Ex\bracket{k_r} = O(K \cdot p_r^2)$. 
\end{lemma}
\begin{proof}
	We prove the lemma by induction on $r$. For the base case, recall that a level-$1$ phase $P^1_i$ is successful iff $\card{P^1_i} = T_1 (= \frac{1}{24p_1^2})$. The probability that $P^1_i$ is 
	successful is at least $1/2$ by Lemma~\ref{lem:is-successful-n1/2}. Any successful phase includes $T_1$ updates and hence we can have at most $K/T_1$ successful level-$1$ phases (even if we assume the other phases include no updates). 
	By the same argument as in Lemma~\ref{lem:expected-p}, we have that $\Ex\bracket{k_1} \leq 2K/T_1 = O(K \cdot p_1^2)$. 
	
	We now prove the induction step. Recall that a level-$r$ phase $P^r_i$ is successful iff the level-$(r-1)$ phase that contains it terminate, or $\card{P^r_i} = T_r (= \frac{1}{24p_r^2})$. The probability of being successful 
	is also at least $1/2$ by Lemma~\ref{lem:is-successful-n1/2}. Finally, note that at most $k_{r-1}$ level-$r$ phases can terminate because the corresponding level-$(r-1)$ phase that contain them terminated (by definition of $k_{r-1}$). 
	The number of remaining successful phases are at most $K/T_r$. As such, by the above argument $\Ex\bracket{k_r} \leq 2K/T_r + \Ex\bracket{k_{r-1}} = O(K \cdot p_r^2)$ by induction hypothesis as $p_r \geq 2\cdot p_{r-1}$. 
\end{proof}

As such, the expected running time of this part is:
\[ 
O(K \cdot n \cdot p_r^2 \cdot \Delta_{r-1} + K) = O(K \cdot \log^{2}{n}) \cdot \paren{n \cdot \frac{\Delta_{r-1}}{\Delta_r^2}},
\] 
by the choice of $\Delta_r$ and $p_r$ (note that $n \cdot \Delta_{r-1} > \Delta_r^2$ for all $r \in [R]$).

\paragraph{Total Time Spent for Maintaining $\mis^*_t$.} By Lemma~\ref{lem:update-time-mis}, the total time spent for maintaining $\mis^*_t$ is $O(k_R \cdot n \cdot \Delta_R + K \cdot \Delta_R)$. As, by Lemma~\ref{lem:expected-k_r},
$\Ex\bracket{k_R} = O(K \cdot p_R^2)$, we have that the expected running part of this time is: 
\[
O(K \cdot p_R^2 \cdot n \cdot \Delta_R + K \cdot \Delta_R) = O(K \cdot \log^2{n}) \cdot \Paren{\frac{n}{\Delta_R} + \Delta_R}.
\] 

\paragraph{Total Running Time.} The total expected running time of the algorithm is now: 
\begin{align*}
	O(K \cdot \log^{2}{n}) \cdot \Paren{\sum_{r=1}^{R} \paren{\frac{n \cdot \Delta_{r-1}}{\Delta_r^2 \cdot }} + \frac{n}{\Delta_R} + \Delta_R}.
\end{align*}

Recall that $\Delta_r := 5p_r^{-1}\cdot\ln{n}$. We pick the values of $\Delta_1,\ldots,\Delta_R$ (by choosing $p_1,\ldots,p_R$ in the algorithm) to optimize the above bound. By our assumption that $p_R > 2p_{R-1}$, 
we have that $\frac{n \cdot \Delta_{R-1}}{\Delta_R^2} > \frac{n}{\Delta_R}$. As such we can simplify the bound above to: 
\begin{align*}
	O(K \cdot \log^{2}{n}) \cdot \Paren{\sum_{r=1}^{R} \paren{\frac{n \cdot \Delta_{r-1}}{\Delta_r^2 }} + \Delta_R}.
\end{align*}

To optimize this bound, we form the following equations: 
\begin{align}
	\frac{n^2}{\Delta_1^2} = \frac{n\Delta_1}{\Delta_2^2} = \frac{n\Delta_2}{\Delta^2_3} = \ldots = \frac{n\Delta_{R-1}}{\Delta_R^2} = \Delta_R. \label{eq:good-eq}
\end{align}

One can then use all the equalities except for the last one to prove by induction that: 
\begin{align*}
	\Delta_{i} = \Delta_{i+1}^{\paren{\frac{2^{i+1}-2}{2^{i+1}-1}}} \cdot n^{\paren{\frac{1}{2^{i+1}-1}}}. 
\end{align*}

Then using the final equality in Eq~(\ref{eq:good-eq}), we obtain that: 
\begin{align*}
	\Delta_{R} = n^{\paren{\frac{1}{2} \cdot \frac{2^{R}}{2^R-1}}} = O(\sqrt{n}), 
\end{align*}
where the second inequality is by the choice of $R = 2\log\log{n}$, and thus having  $n^{\paren{\frac{1}{2^{R}-1}}} = O(1)$\footnote{We remark our algorithm in Section~\ref{sec:dynamic-n2/3} can be seen as a special case of the algorithm in this section
with parameter $R=2$ instead of $R = 2\log\log{n}$. Using the calculation above, it is easy to see that for the choice of $R = 2$, the bound on $\Delta_R$ is $O(n^{2/3})$ as in the algorithm in Section~\ref{sec:dynamic-n2/3} (note that we are measuring the $\log{n}$-parameters outside this calculation).}. 

All in all, this implies that the total expected running time of the algorithm is: 
\begin{align*}
	O(K \cdot \log^{2}{n}) \cdot \Paren{R \cdot \sqrt{n}} = O(K \cdot \sqrt{n} \cdot \log^{2}{n} \cdot \log\log{n}),
\end{align*}
finalizing the proof of expectation-bound in Theorem~\ref{thm:dynamic-n1/2}. 

To obtain the bound, with high probability, we can apply the same exact argument in Lemma~\ref{lem:high-probability-p} in Section~\ref{sec:dynamic-n2/3} to Lemma~\ref{lem:expected-k_r}; as the smallest value of $p_r$ for $r \in [R]$ belongs to $p_1$
and it is equal to $\Theta(n^{-3/4} \cdot \log{n})$ (simply plug in the value of $\Delta_R$ in the first term of Eq~(\ref{eq:good-eq})), we obtain that as long as $K = \Omega(n^{3/2}\log{n})$, we obtain the bound with high probability. For smaller values of $K$, 
we again do as in Section~\ref{sec:dynamic-n2/3}, and obtain that the total running time of the algorithm in this case is $O(n^{2}\log^{3}{n})$, concluding the proof of Theorem~\ref{thm:dynamic-n1/2}.

%\clearpage

%\input{dynamic-delta}
%
%\input{dynamic-m}

%\input{distributed-m}

%\input{distributed-m}

%\subsection*{Acknowledgements}

%\clearpage

\bibliographystyle{abbrv}
\bibliography{randomMMbibfile}

%\clearpage
%\appendix
%

\end{document}